\documentclass[11pt]{article}
\usepackage[margin=1.0in]{geometry}

\usepackage{amsmath}
\usepackage{amssymb}
\usepackage{graphicx}
\usepackage{enumerate}
\usepackage{amsthm}
\usepackage{amsfonts}
\usepackage[all]{xy}
\usepackage{array}
\usepackage{datetime}
\usepackage{centernot}
\usepackage{tikz}
\usepackage[ruled,lined,linesnumbered]{algorithm2e}

\newcommand{\Z}{\mathbb{Z}}
\newcommand{\Q}{\mathbb{Q}}

\newcommand{\divs}{\mid}

\newcommand{\vect}[1]{\langle #1 \rangle}

\newcommand{\den}{\textnormal{den}}

\newcommand{\iterres}{\textnormal{iterres}}

\newcommand{\cgcd}{\textnormal{c-gcd}}

\newcommand{\mdeg}{\textnormal{mdeg}}

\newcommand{\res}{\textnormal{res}}

\newcommand{\lc}{\textnormal{lc}}
\newcommand{\denom}{\textnormal{denom}}

\newcommand{\mvar}{\textnormal{mvar}}

\newtheoremstyle{solutionstyle} 
    {\topsep}                    
    {\topsep}                    
    {}                           
    {}                           
    {\it}                   
    {:}                          
    {.5em}                       
    {}  

\theoremstyle{solutionstyle}

\newtheoremstyle{exercisestyle} 
    {\topsep}                    
    {\topsep}                    
    {}                           
    {}                           
    {\bf}                   
    {:}                          
    {.5em}                       
    {}  

\theoremstyle{exercisestyle}

\newdateformat{mydate}{{\THEDAY} {\monthname[\THEMONTH]} \THEYEAR}

\theoremstyle{definition}
\newtheorem{theorem}{Theorem}
\newtheorem{lemma}{Lemma}
\newtheorem{prop}{Proposition}

\newtheorem{corollary}{Corollary}

\theoremstyle{remark}
\newtheorem*{defn}{Definition}
\newtheorem{ex}{Example}

\begin{document}

\title{Resolving zero-divisors using Hensel lifting}
\date{}
\author{John Kluesner and Michael Monagan}
\maketitle

\begin{center}
Department of Mathematics, Simon Fraser University \\
Burnaby, British Columbia, V5A-1S6, Canada \\
\verb|jkluesne@sfu.ca  |  \verb|  mmonagan@sfu.ca|
\end{center}

%
\vspace{30pt}

\begin{abstract}
Algorithms which compute modulo triangular sets must respect the presence of zero-divisors. 
We present Hensel lifting as a tool for dealing with them. 
We give an application: a modular algorithm for computing GCDs of univariate polynomials with 
coefficients modulo a radical triangular set over $\Q$. 
Our modular algorithm naturally generalizes previous work from algebraic number theory. 
We have implemented our algorithm using Maple's
\textsc{recden} package. We compare our implementation with the procedure 
{\tt RegularGcd} in the {\tt RegularChains} package. 
\end{abstract}

\section{Introduction}

Suppose that we seek to find the greatest common divisor of two polynomials $a,b\in \Q(\alpha_1,\dots,\alpha_n)[x]$ where $\alpha_i$ are algebraic numbers. This problem was first solved using a modular algorithm by Langemyr and McCallum \cite{langmccal} and improved by Encarnacion \cite{enc}. Their solution first found a primitive element and then applied an algorithm for one extension. Monagan and van Hoeij \cite{hoeij} improved the multiple extension case by circumventing the primitive element.

The computational model used for an algebraic number field is the quotient ring $\Q[z_1,\dots,z_n]/T$ where $T = \vect{t_1(z_1),t_2(z_1,z_2),\dots,t_n(z_1,\dots,z_n)}$ and each $t_i$ is the minimal polynomial of $\alpha_i$,
hence irreducible, over $\Q(\alpha_1,\dots,\alpha_{i-1})$. 
A natural generalization, requested by Daniel Lazard at ISSAC 2002,
is to consider the same problem when each $t_i$ is possibly reducible
in which case $\Q[z_1,\dots,z_n]/T$ has zero-divisors.

The generators of $T$ form what is known as a triangular set. Let $R = \Q[z_1,\dots,z_n]/T$. This paper proposes a new algorithm for computing $\gcd(a,b)$ with $a,b\in R[x]$. The backbone of it is the Euclidean algorithm. However, the EA can't always be used in this ring. For example, suppose $R = \Q[z_1,z_2]/ T$ and $T = \vect{z_1^2+1, z_2^2 + 1}$. Notice that $z_1^2-z_2^2 = 0$ in $R$ hence
$z_1 - z_2$ and $z_1+z_2$ are zero-divisors in $R$. Consider computing the gcd of
\begin{align*}
a &= {x}^{4}+ \left( z_1+18\,z_2 \right) {x}^{3}+
 \left( - z_2+3\,z_1 \right) {x}^{2}+324\,x +323\\
b &= {x}^{3}+ \left( z_1+18\,z_2 \right) {x}^{2}+
 \left( -19\,z_2+2\,z_1 \right) x+324
\end{align*}
using the Euclidean algorithm. The remainder of $a \div b$ is
$$
r_1 = (z_1+18z_2)x^2 + 323.
$$
Since $z_1 + 18z_2$ is a unit, a division can be performed; dividing $b$ by $r_1$ gives
$$
r_2 = (z_1 - z_2)x + 1.
$$
The next step in the Euclidean algorithm would be to invert $z_1-z_2$, but it's a zero-divisor, so it cannot continue.
A correct approach would be to factor $z_2^2+1 = (z_2-z_1)(z_2+z_1) \pmod{z_1^2+1}$ to split the triangular set $T$ into $\{z_1^2+1, z_2-z_1\}$ and $\{z_1^2+1,z_2+z_1\}$. After that, finish the EA modulo each of these new triangular sets. It's possible to combine the results using the Chinese remainder theorem, but that is costly so it is common practice to instead return the output of the EA along with the associated triangular set. For example, see the definition of pseudo-gcd in \cite{hubert} and regular-gcd in \cite{gcdregchains}. We follow this trend with our definition componentwise-gcd in section 4.

Now, consider trying to compute $\gcd(a,b)$ above using a modular GCD algorithm. One would expect to hit the modular image of the same zero-divisor at each prime and hence one could combine them using Chinese remaindering and rational reconstruction. For instance, the EA modulo $13$ will terminate with the zero-divisor $z_1+12z_2 \pmod{13}$ as expected. However, running the EA modulo $17$ terminates earlier because $\lc(r_1) = z_1+18z_2 \equiv z_1+z_2 \pmod{17}$ is a zero-divisor. This presents a problem: $z_1+z_2 \pmod{17}$ and $z_1+12z_2\pmod{13}$ will never combine into a zero-divisor no matter how many more
primes are chosen.

To circumvent, our algorithm finds a monic zero-divisor and lifts it using Hensel lifting
to a zero-divisor over $\Q$.  Our technique handles both the expected zero-divisors (such as $z_1+12z_2 \pmod{ 13}$ in the above example) and the unexpected zero-divisors (such as $z_1+z_2 \pmod{ 17}$). A different approach that we tried is Abbott's fault tolerant rational reconstruction as described in \cite{ftrr};  although this is effective, we prefer Hensel lifting as it enables us to split the triangular set immediately thus
saving work.

In section 2, we review important properties of triangular  sets, such as being radical. If $T$ is a radical triangular set over $\Q$, reduction modulo $p$ doesn't always result in a radical triangular set. We prove that if $T$ is radical over $\Q$, then $T\mod p$ is radical for all but finitely many primes. We give an algorithm for determining if a prime $p$ enjoys this property, which is based on a corollary from Hubert \cite{hubert}.

In section 3, we present how to use Hensel lifting to solve the zero-divisor problem. We prove a variant of Hensel's lemma that's applicable to our ring and give explicit pseudo-code for a Hensel lifting algorithm. The algorithm is chiefly the Hensel construction, but the presence of zero-divisors demands a careful implementation.

In section 4, we give an application of the Hensel lifting to a modular gcd algorithm. Here, we define componentwise-gcds and prove they exist when $T$ is a radical triangular set. We handle bad and unlucky primes, as par for the course with any modular algorithm. Our algorithm is best seen as a generalization of Monagan and van Hoeij's modular gcd algorithm over number fields \cite{hoeij}. We give pseudo-code for the modular gcd algorithm and all necessary sub-procedures. A second application which we are currently 
exploring is the inversion problem, that is, given $u\in \Q[z_1,\dots,z_n]/T$, determine if $u$ is invertible and if so
compute $u^{-1}$.

In section 5, we discuss our implementation of the previously described algorithms in Maple using Monagan and van Hoeij's \textsc{recden} package which uses a recursive dense data structure for polynomials and algebraic extensions. We compare it with the {\tt RegularGcd} procedure in Maple's {\tt RegularChains} package, which uses the subresultant algorithm of Li, Maza, and Pan as described in \cite{gcdregchains}. This comparison includes examples and time tests.

In section 6, we give a complexity analysis of our modular gcd algorithm. This involves a new result about the number of operations it takes to multiply $a,b\in R$ and reduce by $T$. We conclude with expected and worse case running time of our modular gcd algorithm.
We end with a conclusion in section 7.

\section{Triangular Sets}

\subsection{Notation and Definitions}

We begin with some notation. All computations will be done in the ring $k[z_1,\dots,z_n]$ endowed with the monomial ordering $z_i < z_{i+1}$ and $k$ a field.
Let $f\in k[z_1,\dots,z_n]$ be non-constant. The {\it main variable} $\mvar(f)$ of $f$ is the largest variable with nonzero degree in $f$, and the {\it main degree} of $f$ is $\mdeg(f) = \deg_{\mvar(f)}(f)$.

As noted in the introduction, triangular sets will be of key interest in this paper. Further, they are to be viewed as a generalization of an algebraic number field with multiple extensions. For this reason, we impose extra structure than is standard:

\begin{defn}
 A {\it triangular set} $T$ is a set of non-constant polynomials in $k[z_1,\dots,z_n]$ satisfying
 \begin{itemize} \item[]
 (i)  $|T| = n$, \\
 (ii)  $T = \{t_1,\dots,t_n\}$ where $\mvar(t_i) = z_i$, \\
 (iii)  $t_i$ is monic with respect to $z_i$, and \\
 (iv)  $\deg_{z_j}(t_i) < \mdeg(t_j)$ for $j < i$.
 \end{itemize}
The degree of $T$ is $\prod_{i=1}^n \mdeg(t_i)$. Also, $T = \emptyset$ is a triangular set.
\end{defn}
Condition (i) states there are no unused variables. This is equivalent to $T$ being zero-dimensional. Condition (ii) gives a standard notation that will be used throughout this paper. Conditions (iii) and (iv) relates the definition to that of minimal polynomials. Condition (iv) is commonly referred to as a reduced triangular set as seen in \cite{theoryoftrisets}. The degree of $T$ is akin to the degree of an extension.

\begin{ex}
The polynomials $\{z_1^3 + 4z_1, z_2^2 + (z_1+1)z_2 + 4\}$ form a triangular set. However, $\{z_2^2 + (z_1+1)z_2 + 4\}$ wouldn't  since there's no polynomial with $z_1$ as a main variable. Also, 
$\{ t_1 = z_1^3 + 4z_1, ~ t_2 = z_2^2 + z_1^4z_2 + 3 \}$ isn't because $\deg_{z_1}(t_2) = 4 > \mdeg(t_1)$.
\end{ex}

A zero-divisor $u\in k[z_1,\dots,z_n]$ modulo $T$ is a polynomial such that $u\not\in\vect{T}$ and there is a polynomial $v \not\in \vect{T}$ where $uv \in \vect{T}$. Since $R$ is a finite-dimensional $k$-algebra, all nonzero elements are either zero-divisors or units modulo $\vect{T}$.

Given a triangular set $T$, we define $T_i = \{t_1,\dots,t_{i}\}$ and $T_0 = \emptyset$. For example, let $T = \{z_1^3 + 1,\, z_2^3 + 2,\, z_3^3 + 3\}$. Then, $T_3 = T$, $T_2 = \{z_1^3 + 1, z_2^3 + 2\}$, $T_1 = \{z_1^3 + 1\}$. In general, since any triangular set $T$ forms a Grobner basis with respect to the lex monomial ordering,  it follows that $k[z_1,\dots,z_{i}]\cap \vect{T} =\vect{T_{i}}$ when $\vect{T_{i}}$ is viewed as an ideal of $k[z_1,\dots,z_{i}]$;  this is a standard result of elimination theory, see Cox, Little, O'Shea \cite{clo}.

The presence of zero-divisors presents many unforeseen difficulties that the following examples illustrate.

\begin{ex}
It's possible for a monic polynomial to factor as two polynomials with zero-divisors as leading coefficients.
For example, consider the triangular set $T = \{(z_1^2+2)(z_1^2+1), z_2^3-z_2\}$. Observe that when working modulo $(z_1^2+2)(z_1^2+1)$,
$$
z_2^3-z_2 = \left((z_1^2+2)z_2^2 - 1\right)\left((z_1^2+1)z_2^3+z_2\right).
$$
Of course, a nicer factorization may exist, like $z_2^3-z_2 = (z_2^2-1)z$, but it's not clear if this always occurs or how to compute it. This greatly enhances the complexity of handling zero-divisors. The above equation also shows that the degree formula for the product of two polynomials doesn't hold in this setting.
\end{ex}

\begin{ex}
Another difficulty is that denominators in the factors of a polynomial $a(x) \in R[x]$
may not appear in the denominators of $a(x)$.
Weinberger and Rothschild give the following example
in \cite{rothschild}.  Let $t_1(z_1) = z_1^6+3z_1^5+6z_1^4+z_1^3-3z_1^2+12z_1+16$ which
is irreducible over $\Q$.  The polynomial
$$ \textstyle
f = x - \frac{4}{3} - \frac{11}{12} z_1 + \frac{7}{12}z_1^2-\frac{1}{6} z_1^3-\frac{1}{12} z_1^4-\frac{1}{12} z_1^5$$
is a factor of $a(x) = x^3-3$ in $R[x]$.  
The denominator of any factor of $a(x)$ ($\denom(f) = 12$ in this example)
must divide the defect $d$ of the field $R$.  It is known that the discriminant $\Delta$ of $t_1(z_1)$ is
a multiple of $d$, usually, much larger than $d$.
Thus we could try to recover $\Delta f$ with Chinese remaindering
then make this result monic.
Although one could try to generalize the discriminant to the case $n>1$,
using rational number reconstruction circumvents this difficulty and also
allows us to recover $g$ without using a lot more primes than necessary.
\end{ex}

Lastly, since there is no standard definition of $\gcd(a,b)$ for $a,b\in R[x]$ where $R$ is a commutative ring unless $R$ is a unique factorization domain, we'd like to make it explicit that $g = \gcd(a,b)$ if (i) $g\divs a$ and $g\divs b$, and (ii) any common divisor of $a$ and $b$ is a divisor of $g$.

\subsection{Radical Triangular Sets}

An ideal $I\subset k[x_1,\dots,x_n]$ is radical if $f^m\in I$ implies $f\in I$.
To start, we give a structure theorem for radical and zero-dimensional triangular sets. One could prove this more generally by using the associated primes of $T$ as done in Proposition 4.7 of \cite{hubert}. The structure theorem gives many powerful corollaries.

\begin{theorem}
Let $T\subseteq k[z_1,\dots,z_n]$ be a triangular set. Then, $k[z_1,\dots,z_n]/T$ is isomorphic to a direct product of fields if and only if $T$ is zero-dimensional and radical.
\label{thm:gcdsexist}
\end{theorem}

\begin{corollary}
Let $T \subset k[z_1,\dots,z_n]$ be a radical, zero-dimensional triangular set and $R = k[z_1,\dots,z_n]/T$. Let $a,b\in R[x]$. Then a greatest common divisor of $a$ and $b$ exists.
\end{corollary}

\begin{proof}
This follows straightforwardly using the CRT and Theorem \ref{thm:gcdsexist}.
\end{proof}

\begin{corollary}[Extended Euclidean Representation]
Let $T \subset k[z_1,\dots,z_n]$ be a radical, zero-dimensional triangular set and $R = k[z_1,\dots,z_n]/T$. Let $a,b\in R[x]$ with $g = \gcd(a,b)$. Then, there exists polynomial $A,B \in R[x]$ such that $aA + bB = g$.
\label{exteucrep}
\end{corollary}

\begin{proof}
Note that $R \cong \prod F_i$ where $F_i$ is a field, and we can extend this to $R[x] \cong \prod F_i[x]$. Let $a \mapsto (a_i)_i$ and $b \mapsto (b_i)_i$. Define $h_i = \gcd(a_i,b_i)$ in $F_i[x]$. By the extended Euclidean algorithm, there exists $A_i,B_i\in F_i[x]$ such that $a_iA_i + b_iB_i = h_i$. Let $h \mapsto (h_i)_i$ and $A\mapsto (A_i)_i$ and $B\mapsto (B_i)_i$. Clearly, $aA + bB = h$ in $R[x]$. Since $h \divs g$, we can multiply through by the quotient to write $g$ as a linear combination of $a$ and $b$.
\end{proof}

It should be noted that Corollary \ref{exteucrep} works even if running the Euclidean algorithm on $a$ and $b$ encounters a zero-divisor. This shows it's more powerful than the extended Euclidean algorithm. Further, it also applies to the case where $\lc(g)$ is a zero-divisor.

We next turn our attention to working modulo primes.

\begin{defn}
Let $T\subset \Q[z_1,\dots,z_n]$ be a radical triangular set. A prime number $p$ is a radical prime if $p$ doesn't appear as a denominator of any of the polynomials in $T$, and if $T \mod p \subset \Z_p[z_1,\dots,z_n]$ remains radical.
\end{defn}

\begin{ex}
The triangular set $\{z_1^2 - 3\}$ is radical over $\Q$. Since the discriminant of $z_1^2-3$ is 12, it 
follows that $2,3$ aren't radical primes, but all other primes are.
\end{ex}

If there were an infinite family of nonradical primes, it would present a problem for the algorithm. We prove this can't happen. This has also been proven with quantitative bounds in \cite{lifting}. The following lemma is a restatement of Corollary 7.3 of \cite{hubert}. It also serves as the main idea of our algorithm for testing if a prime is radical; see IsRadicalPrime below.

\begin{lemma}
Let $T\subset k[z_1,\dots,z_n]$ be a zero-dimensional triangular set. Then $T$ is radical if and only if $\gcd(t_i,t_i') = 1 \pmod{T_{i-1}}$ for all $i$.
\label{lemma:radicaltest}
\end{lemma}

\begin{theorem}
Let $T\subset \Q[z_1,\dots,z_n]$ be a radical, zero-dimensional triangular set. All but finitely many primes are radical primes.
\end{theorem}

\begin{proof}
By Lemma \ref{lemma:radicaltest}, $\gcd(t_i,t_i') = 1$. By the extended Euclidean representation (Corollary \ref{exteucrep}), there exist polynomials $A_i,B_i \in (\Q[z_1,\dots,z_{i-1}]/T_{i-1})[z_i]$ where $A_it_i + B_it_i' = 1 \pmod{T_{i-1}}$. Take any prime $p$ that doesn't divide the denominator of any $A_i,B_i,t_i,t_i'$. This means one can reduce this equation modulo $p$ and so $A_it_i + B_it_i' \pmod{T_{i-1}, p}$. This implies $\gcd(t_i,t_i') = 1\pmod{T_{i-1},p}$ and so $T$ remains radical modulo $p$ by Lemma \ref{lemma:radicaltest}. There are only a finite amount of primes that divide the denominator of any of these polynomials.
\end{proof}

Lastly, we give an algorithm for testing if a prime $p$ is radical. It may not always output {\tt True} or {\tt False} as it relies on Lemma \ref{lemma:radicaltest} which relies on a gcd computation modulo $p$, which is computed by the Euclidean algorithm. If a zero-divisor is encountered, we output the zero-divisor. This case is caught later in the modular gcd algorithm, of which IsRadicalPrime is a subroutine.

\begin{algorithm}
\caption{IsRadicalPrime\label{IR}}
\SetKwInOut{Input}{Input}\SetKwInOut{Output}{Output}\SetArgSty{text}
\Input{A zero-dimensional, radical triangular set $T\subset \Q[z_1,\dots,z_n]$ and a prime number $p$ that does not divide any denominator of any coefficient of any $t_i\in T$.}
\Output{A boolean indicating if $T$ remains radical modulo $p$, or a zero-divisor.}
\For{$i = 1,\dots,n$}{
    $dt := \frac{\partial}{\partial z_i}T[i]$\;
    $g := \gcd(T[i], dt)$ over $\Z_p[z_1,\dots,z_i]/T_{i-1}$\;
    {\bf if} $g = $ [``ZERODIVISOR'', $u$] {\bf then} \Return [``ZERODIVISOR'', $u$]\;
    {\bf if} $g \neq 1$ {\bf then} \Return{False}\;
}
\Return{True}\;
\end{algorithm}

\section{Handling Zero-Divisors}

We turn our attention to lifting a factorization $f = ab \pmod{T,p}$ for $a,b,f\in R[x]$. A general factorization will not be liftable; certain conditions are necessary for existence and uniqueness of each lifting step. For one, we will need $\gcd(a,b) = 1 \pmod{p}$ as is required in the case with no extensions to satisfy existence. Further, we will need both $a$ and $b$ to be monic to satisfy uniqueness. The following lemma gives a uniqueness criterion for the extended Euclidean representation. It generalizes Theorem 26 in Geddes, Czapor, Labahn \cite{gcl} from $F[x]$ to $R[x]$. We give a proof, but note that it is only a slight alteration.

\begin{lemma}
\label{lemma:diophantsolve}
Let $T \subset k[z_1,\dots,z_n]$ be a zero-dimensional triangular set and $R = k[z_1,\dots,z_n]/T$. Let $a,b\in R[x]$ be nonzero and monic with polynomials $A,B$ where $1 = Aa+Bb$. Then, for any polynomial $c\in R[x]$, there exist unique polynomials $\sigma,\tau\in R[x]$ such that 
$$
a\sigma + b\tau = c, \quad \deg(\sigma) < \deg(b).
$$
\end{lemma}

\begin{proof}
{\it Existence}: Multiplying through $1=Aa+Bb$ by $c$ gives $a(cA) + b(cB) = c$. Dividing $cA$ by $b$, which we can do since $b$ is monic, gives $cA = qb + r$ with $r=0$ or $\deg(r) < \deg(b)$. Define $\sigma = r$ and $\tau = cB+qa$. Observe that
$$
a\sigma + b\tau = ar + b(cB + qa) = ar + bcB + abq = a(r+bq) + bcB = acA + bcB = c(aA+bB) = c
$$
thus $\sigma$ and $\tau$ satisfy the conditions of the Lemma.
%
{\it Uniqueness}: Suppose both pairs $\sigma_1,\tau_1$ and $\sigma_2,\tau_2$ satisfy $a\sigma_i + b \tau_i = c$ with the desired degree constraint. This yields
$$
(\sigma_1 - \sigma_2)a = b(\tau_2 - \tau_1).
$$
Since $\gcd(a,b) = 1$, it follows that $b \divs \sigma_1 - \sigma_2$. However, since $b$ is monic and $\deg(\sigma_1 - \sigma_2) < \deg(b)$, this is only possible if $\sigma_1 - \sigma_2 = 0$. Thus $0 = b (\tau_2 - \tau_1)$.  Next, since $b$ is not a zero-divisor (because it's monic), this can only happen if $\tau_2 - \tau_1 = 0$ as well.
\end{proof}

We're particularly interested in trying to factor $t_n$ modulo $T_{n-1}$ because encountering a zero-divisor may lead to such a factorization; that is, if $w$ is a zero-divisor with main variable $z_n$, we can write $u = \gcd(t_n, w)$ and then $t_n = uv \mod\vect{T_{n-1}}$ by the division algorithm. 
As long as $T$ is radical, the next lemma shows we automatically get $\gcd(u,v) = 1$.

\begin{lemma}
Let $T\subset k[z_1,\dots,z_n]$ be a radical, zero-dimensional triangular set. Suppose $t_n \equiv uv \pmod{T_{n-1}}$. Then, $\gcd(u,v) = 1 \pmod{T_{n-1}}$.
\end{lemma}

\begin{proof}
Let $u = \overline{u}g \pmod{T_{n-1}}$ and $v = \overline{v}g\pmod{T_{n-1}}$. Note that $t_n \equiv \overline{u}\overline{v}g^2 \pmod{T_{n-1}}$. This would imply $(\overline{u}\overline{v}g)^2 \equiv 0 \pmod{T}$; that is, $\overline{u}\overline{v}g$ is a nilpotent element. However, since nilpotent elements don't exist modulo a radical ideal, $\overline{u}\overline{v}g \equiv 0 \pmod{T}$. This would imply $\overline{u}\overline{v}g \equiv q t_n \pmod{T_{n-1}}$ for some polynomial $q$. Then,
$$
(gq-1)t_n \equiv gqt_n-t_n \equiv g\overline{u}\overline{v}g - t_n \equiv 0 \pmod{T_{n-1}}.
$$
Since $t_n$ is monic in $z_n$, it can't be a zero-divisor modulo $T_{n-1}$. Therefore, $gq - 1 \equiv 0 \pmod{T_{n-1}}$. Thus, $g$ is a unit modulo $T_{n-1}$ and so indeed $\gcd(u,v)=1\pmod{T_{n-1}}$.
\end{proof}

Finally, the next proposition shows that lifting is possible. The proof given is simply the Hensel construction.

\begin{prop}
Let $T \subset \Z_p[z_1,\dots,z_n]$ be a zero-dimensional triangular set with $p$ a prime number. Suppose $t_n \equiv u_0v_0 \pmod{T_{n-1},p}$ where $u_0$ and $v_0$ are monic. Assume there are polynomials $A,B$ where $1 = Au_0+Bv_0$. Then, there exist unique monic polynomials $u_k,v_k$ such that $t_n\equiv u_kv_k \pmod{T_{n-1},p^k}$ and 
$u_k \equiv u_0 \mod \pmod{T_{n-1},p}$ and $v_k \equiv v_0 \mod \pmod{T_{n-1},p}$ for all $k \ge 1$.
\end{prop}

\begin{proof}
(by induction on $k$): The base case is clear. For the inductive step, we want to be able to write $u_k = u_{k-1} + p^{k-1} a \pmod{T_{n-1},p^k}$ and $v_k = v_{k-1} + p^{k-1} b \pmod{T_{n-1},p^k}$ satisfying
$$
t_n \equiv u_kv_k \pmod{T_{n-1},p^k}.
$$
Multiplying out $u_k,v_k$ gives
$$
t_n \equiv u_kv_k \equiv u_{k-1}v_{k-1} + p^{k-1}(av_{k-1} + bu_{k-1}) \pmod{T_{n-1},p^k}.
$$
Subtracting $u_{k-1}v_{k-1}$ on both sides and dividing through by $p^{k-1}$ gives
$$
\frac{t_n - u_{k-1}v_{k-1}}{p^{k-1}} \equiv av_0 + bu_0 \pmod{T_{n-1},p}.
$$
Let $c = \frac{t_n-u_{k-1}v_{k-1}}{p^{k-1}}$. By Lemma \ref{lemma:diophantsolve}, there exists unique polynomials $\sigma,\tau$ such that $u_0\sigma + v_0\tau \equiv c \pmod{T_{n-1},p}$ with $\deg(\sigma) < \deg(v_0)$ and $\deg(\tau) < \deg(u_0)$ since certainly 
$
\deg(c) = \deg(t_n - u_{k-1}v_{k-1}) < \deg(t_n) = \deg(u_0) + \deg(v_0).
$
Set $a = \tau$ and $b = \sigma$. Because of these degree constraints, $u_k = u_{k-1} + ap^{k-1}$ has the same leading coefficient as $u_{k-1}$ and hence $u_0$; in particular $u_k$ is monic. Similarly, $v_k$ is monic as well. By uniqueness of $\sigma$ and $\tau$, we get uniqueness of $u_k$ and $v_k$. 
\end{proof}

What follows is a formal presentation of the Hensel construction. 
The algorithm HenselLift takes input $u_0,v_0,f\in R/\vect{p}[x]$ where $u_0,v_0$ are monic and $f = u_0v_0 \pmod{p}$. It also requires a bound $B$ that's used to notify termination of the Hensel construction and output FAIL. A crucial part of the Hensel construction is solving the diophantine equation $\sigma u_0 + \tau v_0 = c \pmod{T,p}$. This is done using the extended Euclidean algorithm and Lemma \ref{lemma:diophantsolve}. It's possible that a zero-divisor is encountered in this process. This has to be accounted for. Therefore, we allow the HenselLift algorithm to also output [``ZERODIVISOR'', $u$] if it encounters a zero-divisor $u\in R/\vect{p}$.

\begin{algorithm}[ht]\label{alg:hensellift}
\caption{HenselLift\label{IR}}
\SetKwInOut{Input}{Input}\SetKwInOut{Output}{Output}\SetArgSty{text}
\Input{A zero-dimensional radical triangular set $T\subset \Q[z_1,\dots,z_n]$, a radical prime $p$, polynomials $f\in R[x]$ and $a_0,b_0\in R/\vect{p}[x]$ where $R = \Q[z_1,\dots,z_n]/ T$, and a bound $B$. Further, assume $f \equiv a_0b_0 \pmod{p}$ and $\gcd(a_0,b_0) = 1$.}
\Output{Either polynomials $a,b\in R[x]$ where $f = ab$, FAIL if the bound $B$ is reached, or [``ZERODIVISOR'', $w$] if a zero-divisor $w\in R/\vect{p}$ is encountered.}
Solve $sa_0+tb_0=1$ using the monic extended Euclidean algorithm for $s,t \in R/\vect{p}[x]$\;
\lIf{a zero-divisor $w$ is encountered}{\Return{[``ZERODIVISOR'', $w$]}}
Initialize $u = a_0, v = b_0$ and lift $u$ and $v$ from $R/\vect{p}$ to $R$\;
\For{ $i=1,2,\dots$ }{
    Set $a := $RationalReconstuction($u \pmod{p^i}$)\;
    \lIf{$a\neq$ FAIL, and $a|f$}{\Return{$a$,$f/a$}}
    \lIf{$p^i > 2 B$}{\Return{FAIL}}
    Compute $e := f - uv$ as polynomials over $\Q$\;
    Set $c := (e/p^i) \mod p$ \;
    Solve $\sigma a_0 + \tau b_0 = c$ for $\sigma,\tau \in R/\vect{p}[x]$ using $s a_0 + t b_0 = 1$\;
    Lift $\sigma$ and $\tau$ from $R/\vect{p}$ to $R$ and set $u := u + \tau p^i$ and $v := v + \sigma p^i$\;
}
\end{algorithm}

In general the input $f$ will have fractions thus the error $e$ in our Hensel lifting algorithm
will also have fractions and hence it can never become 0.  Note the size of the rational
coefficients of $e$ grow linearly with $i$ as $f$ is fixed and the magnitude of the
integer coefficients in the product $uv$ are bounded by $p^{2i} (1+\deg u)$.

The standard implementation of Hensel lifting requires a bound on the coefficients of the
factors of the polynomial $f \in R[x]$.  For the base case $n=0$ where $R[x]=\Q[x]$ 
one can use the Mignotte bound (see \cite{MCA}).
For the case $n=1$ Weinberger and Rothschild \cite{rothschild} give
a bound but note that it is large.  We do not know of any bounds for the general case $n>1$ and hypothesize that they would be bad.
Therefore a more ``engineering''-esque approach is needed.
Since we do not know whether the input zero-divisor $a_0$ is the image of a monic factor of $f$,
we repeat the Hensel lifting each time a zero-divisor is encountered 
in our modular GCD algorithm, first using a bound of $2^{60}$,
then $2^{120}$, then $2^{240}$ and so on, until the coefficients of any monic factor of $f$ can
be recovered using rational number reconstruction.

The prime application of Hensel lifting will be as a solution to the zero-divisor problem. This is the goal of the HandleZeroDivisorHensel algorithm. The algorithm assumes a zero-divisor modulo a prime $p$ has been encountered by another algorithm (such as our modular gcd algorithm). It attempts to lift this zero-divisor using HenselLift. If HenselLift encounters a new zero-divisor $w$, it recursively calls HandleZeroDivisorHensel($w$). If the Hensel lifting fails (i.e., a bound is reached), it instructs the algorithm using it to pick a new prime. If the Hensel lifting succeeds in finding a factorization $t_n = uv \pmod{T_{n-1}}$ over $\Q$, then the algorithm using it works recursively on new triangular sets $T^{(u)}$ and $T^{(v)}$ where $t_n$ is replaced by $u$ and $v$, respectively. 

\begin{algorithm} \label{alg:handlezd:hensel}
\caption{HandleZeroDivisorHensel\label{IR}}
\SetKwInOut{Input}{Input}\SetKwInOut{Output}{Output}\SetArgSty{text}
\Input{A zero-dimensional radical triangular set $T\subset \Z_p[z_1,\dots,z_n]$ modulo a prime $p$ and a zero-divisor $u_0\in R$ where $R = \Z_p[z_1,\dots,z_n]/T$. Assume $\mvar(u) = n$.}
\Output{A message indicating the next steps that should be carried out, including any important parameters\;}
    Set $v_0 := $Quotient($t_n, u_0$)$\pmod{T_{n-1},p}$\;
    \lIf{$v_0 = $ [``ZERODIVISOR'', $w$]}{\Return{HandleZeroDivisorHensel($w$)}}
    {\bf if} the global variable $B$ is unassigned {\bf then }set $B := 2^{60}$ {\bf else} set $B := B^2$\;
    Set $u,v := $HenselLift($t_n,u_0,v_0, B$)\;
    \lIf{$u = $ [``ZERODIVISOR'', $w$]}{\Return{HandleZeroDivisorHensel($w$)}}
    \lElseIf{$u =$ {\it FAIL}}{\Return{\it FAIL}. This indicates that a new prime or bigger bound is needed}
    \lElse{\Return{$u$ and $v$}}
\end{algorithm}

We'd like to make it clear that this is not the first case of using $p$-adic lifting techniques on triangular sets. In particular, lifting the triangular decomposition of a regular chain has been used by Dahan, Maza, Schost, Wu, Xie in \cite{lifting}.

\section{The Modular Algorithm}

The main content of this section is to fully present and show the correctness of our modular algorithm. First, suppose a zero-divisor $w$ over $\Q$ is found while running the modular algorithm. It will be used to factor $t_k = uv \pmod{T_{k-1}}$ where $u$ and $v$ are monic with main variable $z_k$. From here, the algorithm proceeds to split $T$ into $T^{(u)}$ and $T^{(v)}$ where $t_k$ is replaced with $u$ in $T^{(u)}$ and $v$ in $T^{(v)}$. Of course $t_i$ is reduced for $i > k$ as well. The algorithm then continues recursively. Once the recursive calls are finished, we could use the CRT to combine $\gcd$s into a single gcd, but this would be very time consuming. Instead, it's better to just return both gcds along with their associated triangular sets. This approach is similar to Hubert's in \cite{hubert} which she calls a pseudo-gcd. Here, we refer to this as a component-wise gcd, or c-gcd for short:

\begin{defn}
Let $R$ be a commutative ring with unity such that $R \cong \prod_{i=1}^r R_i$ and $a,b\in R[x]$. Let $\pi_i\colon R \to R_i$ be the natural projections. A component-wise gcd of $a$ and $b$ is a tuple $(g_1,\dots,g_r)\in\prod_{i=1}^r R_i[x]$ where each $g_i = \gcd(\pi_i(a),\pi_i(b))$ and $\lc(g_i)$ is a unit.
\end{defn}

The modular algorithm's goal will be to compute $\cgcd(a,b)$ given $a,b\in R[x]$ where $R = \Q[z_1,\dots,z_n]/T$ and $T\subset \Q[z_1,\dots,z_n]$ is a radical triangular set.
As with all modular algorithms, it's possible that some primes are unlucky. We also prove this only happens for a finite number of cases.

\begin{defn}
Let $T\subset \Q[z_1,\dots,z_n]$ be a radical triangular set, and $R = \Q[z_1,\dots,z_n]/T$. Let $a,b\in R[x]$ and $g = \cgcd(a,b)$. A prime number $p$ is an {\it unlucky prime} if $g$ doesn't remain a componentwise greatest common divisor of $a$ and $b$ modulo $p$. 
Additionally, a prime is {\it bad} if it divides any denominator in $T$, any denominator in $a$ or $b$, or if $\lc(a)$ or $\lc(b)$ vanishes modulo $p$.
\end{defn}

\begin{theorem}
Let $T\subset \Q[z_1,\dots,z_n]$ be a radical triangular set, and $R = \Q[z_1,\dots,z_n]/T$. Let $a,b\in R[x]$ and $g = \cgcd(a,b)$. Only finitely many primes are unlucky.
\label{thm:unluckyprimes}
\end{theorem}

\begin{proof}
Let $R[x] \cong \prod R_i[x]$ where $g = (g_i)$ and $g = \gcd(a,b)\in R_i[x]$. Let $a \mapsto (a_i)$ and $b \mapsto (b_i)$. If $g_i = 0$, then $a_i = 0$ and $b_i = 0$ and no primes are unlucky since, $\gcd(0,0) \equiv 0 \pmod{p}$. Suppose $g_i = \gcd(a_i,b_i)$ is nonzero and monic. Let $\overline{a}_i$ and $\overline{b}_i$ be the cofactors $a_i = g_i \overline{a}_i$ and $b_i = g_i \overline{b}_i$. I claim $\gcd(\overline{a}_i, \overline{b}_i) = 1$. To show this, consider a common divisor $f$ of $\overline{a}_i$ and $\overline{b}_i$. Note that $fg_i \divs a_i$ and $fg_i \divs b_i$. Since $g_i = \gcd(a_i,b_i)$, it follows that $fg_i \divs g_i$; so there exists $q \in R_i[x]$ where $fg_iq = g_i$. Rewrite this equation as $(fq-1)g_i = 0$. Well, $g_i$ is monic in $x$, and so can't be a zero-divisor. This implies $fq - 1 = 0$ and so indeed $f$ is a unit. Thus, $\gcd(\overline{a}_i,\overline{b}_i) = 1$. By the extended Euclidean representation (Corollary \ref{exteucrep}), there exists $A_i,B_i\in R_i[x]$ where $\overline{a}_iA_i + \overline{b}_iB_i = 1$.

Let $p$ be a prime where $p$ doesn't divide any of the denominators in $a_i,\overline{a}_i,A_i,b_i,\overline{b}_i,B_i,g_i$. Then, we can reduce the equations 
\begin{gather}
\overline{a}_iA_i + \overline{b}_iB_i = 1 \pmod{p},\\
a_i = g_i \overline{a}_i \pmod{p}, \qquad b_i = g_i \overline{b}_i \pmod{p}.
\end{gather}
We will now show that $g_i = \gcd(a_i,b_i) \pmod{p}$. By (2), we get $g_i$ is a common divisor of $a_i$ and $b_i$ modulo $p$. Consider a common divisor $c$ of $a_i$ and $b_i$ modulo $p$. Multiplying equation (1) through by $g_i$ gives $a_iA_i + b_iB_i = g_i \pmod{p}$. Clearly, $c \divs g_i$ modulo $p$. Thus, $g_i$ is indeed a greatest common divisor of $a_i$ and $b_i$ modulo $p$. As there are finitely many primes that can divide the denominators of fractions in the polynomials $a_i,\overline{a}_i,A_i,b_i,\overline{b}_i,B_i,g_i$, there are indeed finitely many unlucky primes.
\end{proof}

\begin{ex}
This example illustrates how the IsRadical function can run into a zero-divisor. 
Consider $T = \{z_1^2-1, z_2^3 + 9z_2^2 + \frac{3z_1+51}{2}z_2 - \frac{53z_1+3}{2}\}$. We will be running the algorithm over $\Q$ to illustrate. First, it would determine that $T_1 = \{z_1^2-1\}$ is radical. Now, when it is running the Euclidean algorithm on $t_2 = z_2^3 + 9z_2^2 + \frac{3z_1+51}{2}z_2 - \frac{53z_1+3}{2}$ and $t_2' = 3z_2^2 + 18z_2 + \frac{3z_1+51}{2}$, the first remainder would be $(z_1 - 1)z_2 -28z_1-27$. However, $z_1 - 1$ is a zero-divisor, so the algorithm would output [``ZERODIVISOR'', $z_1 - 1$]. This same zero-divisor will show up for every odd prime ($2$ appears in the denominator of $t_2$ and so shouldn't be considered). This explains why we can't just simply pick a new prime in
ModularC-GCD if IsRadical encounters a zero-divisor.
\end{ex}

We would like to give a high level overview of the algorithm since looking at pseudo-code is not always the best way to understand. Please see Algorithm \ref{alg:modcgcd} for pseudo-code. The inputs are $a,b\in R[x]$ where $T$ is a radical triangular set and $R = \Q[z_1,\dots,z_n]/T$,
\begin{enumerate}
\item Pick a new prime $p$ that is not bad.
\item Test if $p$ is a radical prime.
      \begin{enumerate}[2.1]
      \item If a zero-divisor is encountered, resolve it using HandleZeroDivisorHensel.
      \item If $p$ is not radical,  go back to step 1. Otherwise, continue as $p$ is a radical prime.
      \end{enumerate}
\item Use the monic Euclidean algorithm to compute $g_p = \gcd(a,b) \pmod{p}$.
\begin{enumerate}[3.1]
      \item If a zero-divisor is encountered, resolve it using HandleZeroDivisorHensel.
      \item Combine all gcds computed modulo primes of lowest degree using Chinese remaindering and rational reconstruction into a polynomial $h$ over $\Q$.
\item Test if $h\divs a$ and $h\divs b$. If the division test succeeds, return $h$. Otherwise, we need more primes, so go back to step 1.
      \end{enumerate}
\end{enumerate}

\noindent
The crux of ModularC-GCD is an algorithm to compute $\gcd(a,b)$ for the input polynomials $a,b$ reduced modulo a prime.
The algorithm we use for this is MonicEuclideanC-GCD below. It is a variant of the monic Euclidean algorithm.
For computing inverses, the extended Euclidean algorithm can be used; modifying MonicEuclideanC-GCD to do this
is straightforward. 

\begin{algorithm}[ht!]\label{alg:moniceuc}
\caption{MonicEuclideanC-GCD\label{IR}}
\SetKwInOut{Input}{Input}\SetKwInOut{Output}{Output}\SetArgSty{text}
\Input{A ring $R$ as specified in the opening of the section, and two polynomials $a,b\in R[x]$. Assume $\deg_x(a) \geq \deg_x(b)$.}
\Output{Either monic $\gcd(a,b)$ or an error if a zero-divisor is encountered.}
\If{$b=0$}{
   \lIf{$\lc(a)$ is a zero-divisor}{\Return{[``ZERODIVISOR'', $\lc(a)$]}}
   \Return{$\lc(a)^{-1}a$}
}
Set $r_0 := a$ and $r_1 := b$\;
$i := 1$\;
\While{$r_{i} \neq 0$}{
    \lIf{$\lc(r_{i})$ is a zero-divisor}{\Return{[``ZERODIVISOR'', $\lc(r_{i})$]}}
    $r_i := \lc(r_i)^{-1}r_i$\;
    Set $r_{i+1}$ as the remainder of $r_{i-1}$ divided by $r_i$\;
    $i := i + 1$\;
}
\Return{$r_{i-1}$}
\end{algorithm}

\newcommand\mycommfont[1]{\small\ttfamily\textcolor{black}{#1}}
\SetCommentSty{mycommfont}

A short discussion about the zero-divisors that may appear is warranted. To compute an inverse, the modular algorithm will be using the extended Euclidean algorithm. The first step would be to invert a leading coefficient $u$ of some polynomial. This requires a recursive call to ExtendedEuclideanC-GCD($u, t_k$)$\pmod{T_{k-1}}$ where $z_k = \mvar(u)$. If $u$ isn't monic, then it would again attempt to invert $\lc(u)$. Because of the recursive nature, it will keep inverting leading coefficients until it succeeds or a monic zero-divisor is found. The main point is that we may assume that the zero-divisors encountered are monic.

Now that all algorithms have been given, we give a proof of correctness for ModularC-GCD. First, we show that a finite number of zero-divisors can be encountered. This ensures that the algorithm terminates.  After that, we prove a lemma about the primes that may occur in a monic factorization modulo the triangular set; note this is nontrivial by example 3. This a key step in the proof that the returned value of ModularC-GCD is correct. The proof will require the concept of localization, the formal process of including denominators in a ring; see Bosch \cite{bosch} for details. For notation purposes, we let $S$ be a set of prime numbers and define $R_S$ as the localization of $R$ with respect to $S$. Note that when $R = \Q[z_1,\dots,z_n]/T$, it's required that any prime dividing any $\den(t_i)$ must be included in $S$ for $R_S$ to be a ring. We will also need the concept of the iterated resultant. Given a triangular set $T$, the iterated resultant of $f$ with $T$ is
$$
\iterres(f,T) = \iterres(\res(f,t_n), T_{n-1}), \qquad \iterres(f,\{t_1\}) = \res(f,t_1).
$$
One important property is that if $f,T\in R'[x] \subset R[x]$ where $R'$ is a subring, then there exist $A,B_1,\dots,B_n\in R'[x]$ where $Af + B_1t_1 + \cdots + B_nt_n = \iterres(f,T)$. This follows from the same proof as given in Theorem 7.1 of \cite{gcl}. Another important property is that $\iterres(f,T) = 0$ if and only if $f$ is a zero-divisor, see \cite{theoryoftrisets}.

\begin{prop}
Let $R = \Q[z_1,\dots,z_n]/T$ where $T$ is a radical zero-dimensional triangular set. Put $a,b\in R[x]$. A finite number of zero-divisors are encountered when running ModularC-GCD($a,b$).
\end{prop}

\begin{proof}
We use induction on the degree of the extension $\delta = d_1\cdots d_n$ where $d_i = \mdeg(t_i)$. If $\delta = 1$, then $R=\Q$ so no zero-divisors occur.

First, there are a finite number of non-radical primes. So we may assume that $T$ remains radical modulo any chosen prime.
Second, consider (theoretically) running the monic Euclidean algorithm over $\Q$ where we split the triangular set if a zero-divisor is encountered. In this process, a finite number of primes divide either denominators or leading coefficients; so we may assume the algorithm isn't choosing these primes without loss of generality.

Now, suppose a prime $p$ is chosen by the algorithm and a zero-divisor $u_p$ is encountered modulo $p$ at some point of the algorithm. 
This implies $\gcd(u_p, t_k) \not\equiv 1 \pmod{T_{k-1}, p}$. We may assume that $u_p = \gcd(u_p,t_k) \pmod{T_{k-1}, p}$ and that $u_p$ is monic; this is because the monic Euclidean algorithm will only output such zero-divisors. If $u_p$ lifts to a zero-divisor over $\Q$, the algorithm constructs two triangular sets, each with degree smaller than $\delta$. So by induction, a finite number of zero-divisors occur in each recursive call. Now, suppose lifting fails. This implies there is some polynomial $u$ over $\Q$ that reduces to $u_p$ modulo $p$ and appears in the theoretical run of the Euclidean algorithm over $\Q$. Note that $\gcd(u,t_k) = 1 \pmod{T_{k-1}}$ over $\Q$ since we're assuming the lifting failed. By Theorem \ref{thm:unluckyprimes}, this happens for only a finite amount of primes. Thus, a finite number of zero-divisors are encountered.
\end{proof}

\begin{algorithm}[ht!]\label{alg:modcgcd}
\caption{ModularC-GCD\label{IR}}
\SetKwInOut{Input}{Input}\SetKwInOut{Output}{Output}\SetArgSty{text}
\Input{A zero-dimensional, radical triangular set $T\subset \Q[z_1,\dots,z_n]$ and two polynomials $a,b\in R[x]$ where $R = \Q[z_1,\dots,z_n]/T$. Assume $\deg(a) \geq \deg(b) \geq 0$.}
\Output{A tuple consisting of comaximal triangular sets $T^{(i)}$ such that $T = \bigcap T^{(i)}$ and 
$g_i = \gcd(a,b) \mod \langle T^{(i)} \rangle$ where $g_i = 0$ or $\lc(g_i)$ is a unit.}
Initialize $dg := \deg(b)$, $M = 1$\;
{\bf Main Loop:} Pick a prime $p$ that is not bad; 
Test if $p$ is a radical prime, $N := $ isRadicalPrime($T,p$)\;
\uIf{$N = $ [``ZERODIVISOR'', $u$]}{
    $K :=$ HandleZeroDivisorHensel($u$)\;
    \lIf{$K =$ FAIL}{Pick a new prime, go to Main Loop}
    \ElseIf{$K$ is a factorization $t_k = wv \pmod{T_{k-1}}$}{
        Create triangular sets $T^{(w)}$ and $T^{(v)}$ where $t_k$ is replaced by $w$ and $v$, respectively\;
        \Return{ModularC-GCD($a,b$)$ \pmod{T^{(w)}}$, ModularC-GCD($a,b$)$\pmod{T^{(v)}}$}
    }
}
\ElseIf{$N = $ False}{Go to Main Loop\; }
Set $g := \gcd(a,b) \mod \langle T,p \rangle$ using algorithm MonicEuclideanC-GCD\;
\uIf{$g = $ [``ZERODIVISOR'', $u$]}{ $K :=$ HandleZeroDivisorHensel($u$)\;
    \lIf{$K =$ FAIL}{Pick a new prime: Go to Main Loop}
    \ElseIf{$K$ is a factorization $t_k = wv\pmod{T_{k-1}}$}{
        Create triangular sets $T^{(w)}$ and $T^{(v)}$ where $t_k$ is replaced by $w$ and $v$, respectively\;
        \Return{ModularC-GCD($a,b$)$\pmod{T^{(w)}}$, ModularC-GCD($a,b$)$\pmod{T^{(v)}}$}
    }}
\Else
{    
    \uIf{$\deg(g) = dg$}{
        The chosen prime seems to be lucky\;
        Use CRT to combine $g$ with other gcds (if any), store the result in $G$ and set $M := M \times p$\;
    }
    \uElseIf{$\deg(g) > dg$}{
        The chosen prime was unlucky, discard $g$\;
        Pick a new prime: Go to Main Loop\;
    }
    \ElseIf{$\deg(g) < dg$}{
        All previous primes were unlucky, discard $G$\;
        Set $G := g$, $M := p$, and $dg := \deg(g)$\;
    }
    Set $h := $ RationalReconstruction($G \pmod M$)\;
    \lIf{$h \ne {\rm FAIL}$ and $h \divs a$ and $h \divs b$}{\Return $h$}
    Go to Main Loop\;}
\end{algorithm}

\begin{lemma}
Let $T$ be a radical, zero-dimensional triangular set of $F = \Z[z_1,\dots,z_n]$. Suppose $f,u\in R[x]$ are monic such that $u\divs f$. Let 
$$ S = \{\text{prime numbers } p\in\Z : p\text{ is a nonradical prime with respect to $T$}, \text{ or } p\divs \den(f)\}. $$
Then, $u\in F_S[x]/T$. In particular, the primes appearing in denominators of a factorization are either nonradical primes or divisors of $\den(f)$.
\label{lemma:denoffators}
\end{lemma}

\begin{proof}
Proceed by induction on $n$. Consider the base case $n = 1$. Let $t_1 = a_1a_2\cdots a_s$ be the factorization into monic irreducibles. Note that $a_i,a_j$ are relatively prime since $t_1$ is square-free and $a_1,a_2\in F_S$ by Gauss's lemma (since $S$ contains any primes dividing $\den(t_1)$). Let $u_i = u \mod a_i$ and $f_i = f \mod a_i$. By known results from algebraic number theory (see Theorem 3.2 of \cite{enc} for instance), $\den(u_i)$ consists of primes dividing $\Delta(a_i)$ or $\den(f_i)$. Note that any prime $p\divs \Delta(a_i)$ would force $a_i$, and hence $t_1$, to not be square-free modulo $p$. This would imply $p$ is nonradical and so is contained in $S$; in partiulcar, $u_i \in F_S[x]$.

The last concern is if combining $(u_1,u_2,\dots,u_s) \mapsto u$ introduces another prime $p$ into the denominator. We prove this can only happen if $p$ is nonradical. It's sufficient to show that combining two extensions is enough since we can simply combine two at a time until the list is exhausted. Now, consider the resultant $r = \res_{z_1}(a_1,a_2)$. There are polynomials $A,B\in F_S$ where $Aa_1 + Ba_2 = r$. Note that any prime $p\divs r$ forces $\gcd(a_1,a_2) \neq 1 \pmod{p}$ and so $t_1$ wouldn't be square-free; in particular, $A/r,B/r\in F_S$. Now, let $v = (A/r)a_1u_2 + (B/r)a_2u_1$. Note that $v\mod a_1 = (B/r)a_2u_1 = (1-(A/r)a_1)u_1 = u_1$. Similarly, $v\mod a_2 = u_2$. Since the CRT gives an isomorphism, $u = v$ and indeed $u\in F_S[x]$. This completes the base case.

For the general case, we will generalize each step used in the base case. Instead of just factoring $t_1$, we decompose $T$ as a product of comaximal triangular sets known as its triangular decomposition. In place of discriminants of polynomials, we use discriminants of algebraic number fields. Finally, for the combining, iterated resultants are used instead of resultants.

With that in mind, start by decomposing $T$ into its triangular decomposition, which can be done in the the following way:
\begin{enumerate}
\item Factor $t_1 = a_1a_2\cdots a_{s_1}$ into relatively prime monic irreducibles over $\Q$ as in the base case. This gives $\Q[z_1]/T_1$ is isomorphic to the product of fields $\prod_i \Q[z_1]/a_i$. By Gauss's lemma, a prime dividing the $\den(a_i)$ must also divide $\den(f)$. In particular, $a_i\in F_S[x]/T$.
\item We can factor the image of $t_2^{(i)}$ over each $\Q[z_1]/a_i$ into monic relatively prime irreducibles $t_2^{(i)} = b_1^{(i)}b_2^{(i)}\cdots b_{s_2}^{(i)}$. Note that changing rings from $\Q[z_1]/t_1$ to $\Q[z_1]/a_i$ only involves division by $a_i$, and hence the only primes introduced into denominators can come from $\den(a_i)$. 
\item By the induction hypothesis, any prime $p$ dividing $\den(b_j^{(i)})$ is either not a radical prime of the triangular set $\{a_i\}$ or comes from $\den(t_2^{(i)})$. If $\{a_i\}$ isn't radical modulo $p$, then neither is $\{t_1\}$, clearly.
\item Use this to decompose $k[z_1,z_2]/T_2$ into fields $\Q[z_1,z_2]/\vect{a_i, b_j^{(i)}}$ where $a_i,b_j^{(i)}\in F_S[x]/T$.
\item Repeat to get $\Q[z_1,\dots,z_{n}]/T \cong \prod\Q[z_1,\dots,z_{n}]/T^{(i)}$ where each $\Q[z_1,\dots,z_{n}]/T^{(i)}$ is a field and $T^{(i)}\subset F_S$ using the induction hypothesis.
\end{enumerate}
Let $f^{(i)} = f \mod T^{(i)}$ and similarly $u^{(i)} = u \mod T^{(i)}$.
Since $\Q[z_1,\dots,z_{n}]/T^{(i)}$ is an algebraic number field, any prime $p$ occuring in $\den(u^{(i)})$ must either divide the discriminant $\Delta(\Q[z_1,\dots,z_{n}]/T^{(i)})$ or $\den(f^{(i)})$. This implies $p$ must be nonradical with respect to $T^{(i)}$ or divide $\den(f^{(i)})$. (To be more explicit, one could write $\Q[z_1,\dots,z_{n}]/T^{(i)} = \Q(\alpha)$ and note that $p\divs \Delta(\Q[z_1,\dots,z_{n}]/T^{(i)})$ which divides the discriminant $\Delta(m_{\alpha,\Q})$ of the primitive minimal polynomial $m_{\alpha,\Q}$ of $\alpha$. If $p\divs \Delta(m_{\alpha,\Q})$, then $m_{\alpha,\Q}(z)$ isn't square-free and so $\Z_p[z]/m_{\alpha,\Q}$ would contain a nilpotent element.) 

Of course $\den(u^{(i)}) \neq \den(u)$. It remains to show that going from $\prod\Q[z_1,\dots,z_{n}]/T^{(i)}$ to $\Q[z_1,\dots,z_n]/T$ only introduces primes in the denominators that are divisors of $\den(f)$ or nonradical.
This will follow from using iterated resultants similarly to the resultants in the base case. Suppose we are trying to combine $T^{(i)}$ and $T^{(j)}$ with all $t_k^{(i)} = t_k^{(j)}$ besides $t_n^{(i)} \neq t_n^{(j)}$. Now, perform the iterated resultant and write
$$
r = \iterres(\res(t_n^{(i)},t_n^{(j)}),{T_{n-1}^{(i)}}) = At_n^{(i)} + B t_n^{(j)}
$$
with $A,B\in F_S[x]$ since $t_n^{(i)},t_n^{(j)}\in F_S[x]$ are by construction. Well, any prime $p$ that divides $r$ would have the property of $\gcd(t_n^{(i)},t_n^{(j)}) \neq 1 \pmod{p}$. Hence $t_n$ wouldn't be square-free and so $T$ wouldn't be radical mod $p$.
Thus, after recovering all splittings into the ring $\Q[z_1,\dots,z_n][x]/T$, we indeed get $u\in F_S[x]$.
\end{proof}

\begin{theorem}
Let $R = \Q[z_1,\dots,z_n]/T$ where $T$ is a radical zero-dimensional triangular set and let $a,b\in R[x]$. The modular algorithm using Hensel lifting to handle zero-divisors outputs a correct c-gcd if run on $a$ and $b$.
\end{theorem}

\begin{proof}
It is enough to prove this for a single component of the decomposition. For ease of notation, let $T\subset R$ be the triangular set associated to this component. In particular, let $h$ be the monic polynomial returned from the modular algorithm modulo a triangular set $T$ and $g = \gcd(a,b) \pmod{T}$ over $\Q$.
First, we may assume that $b$ is monic. If $\lc_x(b)$ is a unit, divide through by it's inverse and this doesn't change $\gcd(a,b)$. If $\lc_x(b)$ is a zero-divisor, the EA mod $p$ would catch it and cause a splitting, contradicting that the EA mod $p$ didn't encounter a zero-divisor in this component of the $\cgcd$.

Since $h$ passed the trial division in step 34, it follows that $h \divs g$ and hence $\deg(h) \leq \deg(g)$ since $h$ is monic.
Suppose $\lc(g)$ is invertible. If so, make $g$ monic without loss of generality.
Let $p$ be a prime used to compute $h$. Since $g$ is monic and divides $b$ which is also monic, any prime appearing in $\den(g)$ is either nonradical or a divisor of $\den(b)$ by Lemma \ref{lemma:denoffators}. In particular, since the prime $p$ was used successfully to compute $h$, it can't occur in the denominator of $g$. So, we may reduce $g$ modulo $p$. Let $\overline{f}$ denote the reduction of a polynomial $f\in R[x]$ mod $p$. Since $\overline{g}\divs \overline{a}$ and $\overline{g}\divs\overline{b}$, it follows that $\overline{g}\divs \overline{h}$ and so $\deg(g)\leq \deg(h)$. Since $h\divs g$, they have the same degree, and both are monic, it must be that $h = g$ and so indeed $h$ is a greatest common divisor of $a$ and $b$.

Suppose $\lc(g)$ was a zero-divisor and that $\mvar(\lc(g)) = z_n$ without loss. Inspect $\lc_{z_n}(\lc(g))$; if this is a unit, make it monic. If it's a zero-divisor, inspect $\lc_{z_{n-1}}(\lc_{z_n}(g)))$. Continue until $u = \lc_{z_{k+1}}(\cdots (\lc_{z_n}(\lc_x(g))\cdots)$ is a monic zero-divisor. Further, if $\gcd(u,t_{k}) \neq u$, then $u/\gcd(u,t_k)$ is a unit and so we can divide through by it to ensure $\gcd(u,t_k) = u$. Let $t_k = uv \pmod{T_{k-1}}$ be a monic factorization. Note that Lemma \ref{lemma:denoffators} guarantees that the same factorization $\overline{u}\overline{v} = \overline{t_k} \pmod{T_{k-1},p}$ occurs modulo $p$. Hence, we can split $T$ into triangular sets $T^{(u)}$ and $T^{(v)}$ where $t_k$ is replaced by $u$ and $v$, respectively, and this same splitting occurs modulo $p$.

Let $g_u = g \mod T^{(u)}$ and $g_v = g \mod T^{(v)}$ and similarly for other relevant polynomials. It's straightforward to show that $\overline{h_u}$ is still a gcd of $\overline{a_u}$ and $\overline{b_u}$ and $g_u$ for $a_u$ and $b_u$. Now, we consider both triangular sets $T^{(u)}$ and $T^{(v)}$. First, in $T^{(v)}$, $u$ is invertible otherwise $T$ wouldn't be radical. So, multiply $g_v$ by $u^{-1}$ so that $\lc_{z_{k+1}}(\cdots (\lc_{z_n}(\lc_x(g))\cdots) = 1$. Reinspect $w = \lc_{z_{k+2}}(\cdots (\lc_{z_n}(\lc_x(g_v)))\cdots)$. If $w$ isn't a zero-divisor, multiply through by it's inverse and repeat until a zero-divisor is encountered as a leading coefficient. Do the same computations to find another splitting and be in the same situation as that of $u$ in $T$. Otherwise, in $T^{(u)}$, $u = 0$ and so $\lc_{z_{k+1}}(\cdots (\lc_{z_n}(\lc_x(g_u))\cdots)$ has changed; if it's invertible, multiply through by it's inverse until a monic zero-divisor is found in the leading coefficient chain. We again wind up in the situation with a monic factorization of $t_j$ that is reducible modulo $p$.

The process described in the last paragraph must terminate with a splitting in which the image of $g$ is monic since $\lc_x(g)$ has finite degree in each variable. We have already shown that the image of $h$ would be an associate of the image in $g$ in this case. Since being a $\gcd$ persists through isomorphisms, this gives indeed that $h$ is a $\gcd(a,b)$ modulo $T$, as desired.
\end{proof}

\section{Comparison with {\tt RegularGcd}}

We have implemented algorithm {\tt ModularC-GCD} as presented above using 
Maple's \textsc{recden} package which uses a recursive dense data structure
for polynomials with extensions. 
Details can be found in Monagan and van Hoeij's paper \cite{hoeij}.
The reader may find our Maple code for our software there together with several examples and their output
at \verb+http://www.cecm.sfu.ca/CAG/code/MODGCD+.

The remainder of this section will be used to compare our 
algorithm with the {\tt RegularGcd} algorithm (see \cite{gcdregchains})
which is in the {\tt RegularChains} package of Maple.
Algorithm {\tt RegularGcd} computes a subresultant polynomial remainder sequence
and outputs the last non-zero element of the sequence.
We highlight three differences between the output of {\tt RegularGcd} and {\tt ModularC-GCD}.
\begin{enumerate}
\item The algorithms may compute 
different triangular decompositions of the input triangular set.
\item {\tt RegularGcd} returns the last non-zero subresultant but not 
reduced modulo $T$; it often returns a gcd $g$ with $\deg_{z_i}(g) > \mdeg(t_i)$. 
To compute the reduced version, the procedure {\tt NormalForm} is required.
{\tt ModularC-GCD} uses the CRT and rational reconstruction on images of
the $\cgcd$ modulo multiple primes, so it computes the reduced version
of the $\cgcd$ automatically.
\item {\tt RegularGcd} computes gcds up to units, and for some inputs the units can be large.
{\tt ModularC-GCD} computes the monic gcd which may have large fractions. 
\end{enumerate}

\begin{ex}
We'd like to illustrate the differences with an example by an anonymous referee of an earlier version of this paper. 
Let
\begin{align*}
T &= \{x^3-x,~~ {y}^{2}-\tfrac{3}{2}y{x}^{2}-\tfrac{3}{2}yx+y+2{x}^{2}-2\}, \\
a &= z^{2}-\tfrac{8}{3}zy{x}^{2}+3zyx-\tfrac{7}{3}zy-\tfrac{1}{3}z{x}^{2}+3zx-\tfrac{5}{3}z+\tfrac{25}{6}y{x}^{2}-\tfrac{13}{2}yx+\tfrac{10}{3}y+\tfrac{16}{3}{x}^{2}-2x-\tfrac{10}{3}, \\
b &= {z}^{2}+\tfrac{29}{12}zy{x}^{2}+\tfrac{7}{4}zyx-\tfrac{11}{3}zy-\tfrac{8}{3}z{x}^{2}+3zx+\tfrac{2}{3}z+\tfrac{67}{12}y{x}^{2}-\tfrac{11}{4}yx-\tfrac{13}{3}y-\tfrac{13}{3}{x
}^{2}-2x+\tfrac{19}{3}.
\end{align*}
When we run our algorithm to compute $\cgcd(a,b) \pmod{T}$, it returns
\begin{align*}
z^2 + (3x-2)z - 2x+2 &\pmod{y,x^2-1}, \\
z + \tfrac{1}{2}x - \tfrac{3}{2} &\pmod{y - \tfrac{3}{2}x - \tfrac{1}{2},x^2-1}, \\
z+5 &\pmod{y+2,x}, \\
1 &\pmod{y-1,x}.
\end{align*}
The same example using {\tt RegularGcd} returns
\begin{align*}
\left( -96\,y+168 \right) z-552\,y+696 &\pmod{y+2,x},  \\
 154368\,{y}^{3}-117504\,{y}^{2}-559872\,y+585216 &\pmod{y-1,x}, \\
z^2+(\tfrac{2}{3} - \tfrac{8}{3}x^2+3x)z &\pmod{y,x-1}, \\
 (366x^2-90x-96)yz + (102x^2 + 270x-552)y &\pmod{y-2,x-1}, \\
z^2 + (\tfrac{2}{3} - \tfrac{8}{3}x^2+3x)z+\tfrac{19}{13}-\tfrac{13}{3}x^2-2x &\pmod{y,x+1}, \\
 (366x^2-90x-96)yz+(102x^2+270x-552)y &\pmod{y+1,x+1}.
\end{align*}
As can be seen, our algorithm only decomposes $T$ into $4$ triangular sets while {\tt RegularGcd} decomposes $T$ into $6$. Further, it's easy to notice that each component in our output is reduced, while the output of {\tt RegularGcd} isn't. Applying the {\tt NormalForm} command to reduce the output of {\tt RegularGcd} returns
\begin{align*}
360z + 1800 &\pmod{y+2,x}, &
62208 &\pmod{y-1,x}, \\
z^2+z &\pmod{y,x-1}, &
360z-360 &\pmod{y-2,x-1}, \\
z^2-5z+4 &\pmod{y,x+1}, &
-360z + 720 &\pmod{y+1,x+1}.
\end{align*}
Notice that it circumvents fractions. 
In general, the output of our algorithm deals with smaller numbers. This can certainly be seen as an advantage for a user.
\end{ex}


Finally, we'd like to conclude with some timing tests which show the power of
using a modular GCD algorithm that recovers the monic $\cgcd$ from images modulo
primes using rational reconstruction.
We first construct random triangular sets where each $t_i$ is monic
in $z_i$ and dense in $z_1,\dots,z_{i-1}$ with random two digit coefficients.
We then generate $a,b,g\in R[x]$ with degrees $6,5,$ and $4$, respectively.
Then, compute $\cgcd(A,B)$ where $A = ag$ and $B = bg$.
Maple code for generating the test inputs is included on our website.

\begin{table}[!htb]
\begin{center}
\begin{tabular}{cc|rrc|rrc}  \hline
    \multicolumn{2}{c|}{ extension } &  \multicolumn{3}{c|}{ModularC-GCD} & \multicolumn{3}{c}{\tt RegularGcd} \\
    $n$&         degrees  &     time &   divide & \#primes &   time  real &          cpu &  \#terms \\ \hline
     1 &              [4] &    0.013 &    0.006 &        3 &        0.064 &        0.064 &      170 \\
     2 &           [2, 2] &    0.029 &    0.022 &        3 &        0.241 &        0.346 &      720 \\
     2 &           [3, 3] &    0.184 &    0.138 &       17 &         1.73 &        4.433 &     2645 \\
     3 &        [2, 2, 2] &    0.218 &    0.204 &        9 &       10.372 &       29.357 &     8640 \\
     2 &           [4, 4] &    0.512 &    0.391 &       33 &       12.349 &       40.705 &     5780 \\
     4 &     [2, 2, 2, 2] &    1.403 &    1.132 &       33 &      401.439 &      758.942 &   103680 \\
     3 &        [3, 3, 3] &    2.755 &    1.893 &       65 &       413.54 &      1307.46 &    60835 \\
     3 &        [4, 2, 4] &    1.695 &    1.233 &       33 &       39.327 &       86.088 &    19860 \\
     1 &             [64] &    6.738 &    5.607 &       65 &       43.963 &      160.021 &     3470 \\
     2 &           [8, 8] &   13.321 &   11.386 &      129 &      1437.76 &      5251.05 &    30420 \\
     3 &        [4, 4, 4] &   17.065 &   14.093 &      129 &      7185.85 &      22591.4 &   196520 \\ \hline
\end{tabular}

\caption{ \small The first column is the number of algebraic variables,
the second is the degree of the extensions,
the third is the CPU time it took to compute $\cgcd$ of the inputs for ModularC-GCD,
the fourth is the CPU time in ModularC-Gcd spent doing trial divisions over $\Q$,
the fifth is the number of primes needed to recover $g$,
the sixth is the real time it took for {\tt RegularGcd} to do the same computation,
the seventh is the total CPU time it took for {\tt RegularGcd} and the last 
is the number of terms in the unnormalized gcd output by {\tt RegularGcd}. 
All times are in seconds.  }
\end{center}
\end{table}

In the previous dataset, $g$ isn't created as a monic polynomial in $x$,
but ModularC-GCD computes the monic $\gcd(A,B)$. Since $\lc(g)$ is a random polynomial,
its inverse in $R$ will likely have very large rational coefficients, and so additional
primes have to be used to recover the monic gcd. This brings us to an important 
advantage of our algorithm: it is output-sensitive. In Table 2 below $g$ is a monic
degree 4 polynomial with $a$ and $b$ still of degree $6$ and $5$. You'll notice that
our algorithm finishes much faster than the earlier computation, while {\tt RegularGcd}
takes about the same amount of time. This happens because the coefficients of
subresultants of $A$ and $B$ are always large no matter how small the coefficients
of $\gcd(A,B)$ are.

\begin{table}[!htb]
\begin{center}
\begin{tabular}{cc|rrc|rrc}  \hline
    \multicolumn{2}{c|}{ extension } &  \multicolumn{3}{c|}{ModularC-GCD} & \multicolumn{3}{c}{\tt RegularGcd} \\
    $n$&         degrees  &     time &   divide & \#primes &   time  real &          cpu &  \#terms \\ \hline
     1 &              [4] &     0.01 &    0.006 &        2 &        0.065 &        0.065 &      170 \\
     2 &           [2, 2] &     0.02 &    0.016 &        2 &        0.238 &        0.329 &      715 \\
     2 &           [3, 3] &    0.048 &    0.041 &        2 &        1.771 &        4.412 &     2630 \\
     3 &        [2, 2, 2] &     0.05 &    0.041 &        2 &       11.293 &       31.766 &     8465 \\
     2 &           [4, 4] &    0.077 &    0.068 &        2 &       11.521 &       36.854 &     5750 \\
     4 &     [2, 2, 2, 2] &    0.117 &    0.097 &        2 &      321.859 &      431.368 &    99670 \\
     3 &        [3, 3, 3] &    0.222 &    0.201 &        2 &      508.465 &      1615.28 &    57645 \\
     3 &        [4, 2, 4] &     0.05 &    0.032 &        2 &       34.358 &       71.351 &    16230 \\
     1 &             [64] &    0.304 &    0.282 &        2 &        27.55 &       98.354 &     3450 \\
     2 &           [8, 8] &    0.482 &    0.455 &        2 &       1628.7 &      5979.51 &    29505 \\
     3 &        [4, 4, 4] &    0.525 &    0.477 &        2 &      2989.18 &      4751.04 &   192825 \\ \hline
\end{tabular} \\ \medskip
\caption{ \small The columns are the same as for Table 1 }
\end{center}
\end{table}

Let $d_a=\deg_x a$, $d_b=\deg_x b$ with $d_a \ge d_b$ and let $d_g = \deg_x g$.
In Table 3 below we increased $d_a$ and $d_b$ from $6$ and $5$ in Table 1 to $9$ and $8$
leaving the degree of $g$ at 4.
By increasing $d_b$ we increase the number of steps in the Euclidean algorithm
which causes an expression swell in {\tt RegularGcd} in the size of the integer coefficients
and the degree of each $z_1,\dots,z_n$, that is, the expression swell is $(n+1)$ dimensional.
The number of multiplications in $R$ that the monic Euclidean algorithm does is at most
$(d_a-d_b+2)(d_g+d_b)$ for the first division and $\sum_{i=d_g}^{d_g+d_b-1} 2i  = d_b(d_b+2d_g-1)$
for the remaining divisions.  
The trial divisions of $A$ by $g$ and $B$ by $g$ cost $d_a d_g$ and $d_b d_g$ multiplications in $R$ respectively.
Increasing $d_a,d_b,d_g$ from $6,5,4$ in Table 1 to $9,8,4$ increases
the number of multiplications in $R$ in the monic Euclidean algorithm from 87 to 156
and from $24+20=44$ to $36+32=68$ for the trial divisions but the monic gcd 
remains unchanged.  Comparing Table 1 and Table 3 the reader can see that the 
increase in ModularC-GCD is less than a factor of 2.

\begin{table}[!htb]
\begin{center}
\begin{tabular}{cc|rrc|rrc}  \hline
    \multicolumn{2}{c|}{ extension } &  \multicolumn{3}{c|}{ModularC-GCD} & \multicolumn{3}{c}{\tt RegularGcd} \\
    $n$&         degrees  &     time &   divide & \#primes &   time  real &          cpu &  \#terms \\ \hline
     1 &              [4] &    0.021 &    0.011 &        5 &        0.124 &         0.13 &      260 \\ 
     2 &           [2, 2] &    0.043 &    0.031 &        5 &        0.968 &        1.912 &     1620 \\ 
     2 &           [3, 3] &    0.214 &    0.163 &       17 &       10.517 &       34.513 &     6125 \\ 
     3 &        [2, 2, 2] &    0.287 &    0.204 &        9 &       64.997 &       173.53 &    29160 \\ 
     2 &           [4, 4] &    0.638 &    0.427 &       33 &       67.413 &      245.789 &    13520 \\ 
     4 &     [2, 2, 2, 2] &     2.05 &    1.613 &       33 &      2725.13 &      3528.41 &   524880 \\ 
     3 &        [3, 3, 3] &     3.35 &    2.731 &       33 &      3704.61 &      11924.0 &   214375 \\ 
     3 &        [4, 2, 4] &    2.399 &    1.793 &       33 &      334.201 &      869.116 &    68940 \\ 
     1 &             [64] &   10.097 &    8.584 &       65 &      171.726 &      658.518 &     5360 \\
     2 &           [8, 8] &   21.890 &   18.086 &      129 &      10418.4 &      38554.9 &    72000 \\
     3 &        [4, 4, 4] &   37.007 &   31.369 &      129 &     $>50000$ &          --  &      --  \\ \hline
\end{tabular} \\ \medskip
\caption{ \small The columns are the same as for Table 1 }
\end{center}
\end{table}

\section{Complexity Analysis}

We'd like to conclude with a complexity analysis for our algorithm.
Let $R = k[z_1,\dots,z_n]/T$ where $k$ is a field.
To start, we prove a tight bound on the number of field multiplications in $k$
it takes to multiply two polynomials in $R$. 
We assume the inputs are reduced.
We will need this later when doing an asymptotic analysis of the modular algorithms.

Let $\delta$ be the degree of a triangular set $T$ with $n$ variables. 
To multiply two polynomials modulo a triangular set, 
the obvious approach is to multiply out the polynomials and then reduce. 
The reduction step involves doing divisions by the polynomials in the triangular set. 
The way these divisions are done has a large impact on the total number of operations.
We illustrate by describing the classical approach as outlined in \cite{trisetarithmetic}.
We will assume $a$ and $b$ are reduced and dense in all variables.
Let $d_i = \mdeg(t_i)$ for all $i$. 
First, view $a$ and $b$ as polynomials in $z_n$ with coefficients modulo $T_{n-1}$. 
Multiplying $ab$ modulo $T_{n-1}$ involves recursively multiplying all pairs 
of coefficients from $a$ and $b$ and reducing modulo $T_{n-1}$. 
There are $d_n^2$ such pairs and the result is a polynomial $c$ with $\deg_{z_n}(c) = 2(d_n-1)$
with coefficients reduced with respect to $T_{n-1}$.
Next, we have to divide $c$ by $t_n$.
If one uses the high school division algorithm, this involves scaling $d_n$ coefficients 
of $t_n$ for $\deg_{z_n}(c) - d_n+1$ iterations for a total of $d_n(d_n-1)$ recursive multiplications modulo $T_{n-1}$.

Let $M(n)$ be the number of field multiplications used during a multiplication of $a$ and $b$ modulo $T$. The algorithm described above does $d_n^2 + d_n(d_n-1)$ multiplications modulo $T_{n-1}$ each costing $M(n-1)$ field multiplications. This gives a recurrence 
$$
M(n) \leq  (d_n^2 + d_n(d_n-1))M(n-1)
$$
If there are no extensions, it takes a single field multiplication so that $M(0) = 1$. It is straightforward to solve this to get $M(n) = O(2^n\delta^2)$. This is as stated in \cite{trisetarithmetic} for the classical multiplication algorithm. We show that it can in be done in $O(\delta^2)$ field multiplications in Proposition \ref{mulcount:mul}. It should be noted that one normally assumes $\mdeg(t_i) \geq 2$ since extensions by linear polynomials are trivial. With that in mind, the classical multiplication algorithm is $O(\delta^3)$. 
We mention this because it should be clear that Proposition 3 does not turn 
an exponential-time algorithm into a quadratic one, 
but rather a cubic algorithm into a quadratic one.

We would like to note that we have done an actual field multiplication count 
(in our code) and we got the exact same result in the dense case as the proposition states. 
The key idea of the optimization is to do as few recursive reductions as possible.
The idea was originally done by Monagan in \cite{inplace} for the ring $\Z_n$
with $n$ too big for a single machine word.

\medskip\begin{prop} \label{mulcount:mul}
Let $M(n)$ be the number of field multiplications required to multiply $a,b\in k[z_1,\dots,z_n]/T$ and reduce by the triangular set $T$. Let $\mdeg(t_i) = d_i$ and define $\delta_1 = d_1$, $\delta_2 = d_1d_2$, and so on ending with $\delta_n = d_1d_2\cdots d_n = \delta$. Then 
\begin{equation} 
M(n) \leq \delta_n^2 + \sum_{k=1}^n \delta_k^2 \frac{d_k-1}{d_k}\prod_{j=k+1}^n (2d_j-1)
\end{equation}
which is exact in the dense case. Further, $M(n) \leq 3\delta^2$.
\end{prop}

\begin{proof}
Let $D(n)$ be the number of field multiplications it takes to reduce a polynomial of degree $2(d_j-1)$ in each corresponding variable by $T_n$. It is assumed that $D(n)$ works by first reducing by $t_1$, then reducing by $t_2$ modulo $T_1$, etc. Well, multiplying $ab$ will always take $\delta_n^2$ multiplications before reducing, and this is true whether or not these multiplications are done recursively or at the on-set; in particular, it follows that $M(n) = \delta_n^2 + D(n)$. We proceed by describing a division algorithm to divide $c = ab$ by $t_n$ modulo $T_{n-1}$. Let $c = c_0 + c_1z_n + \cdots + c_{2(d_n-1)}z_n^{2(d_n-1)}$ and $t_n = p_0 + p_1z_n + \cdots + z_{d_n}^{d_n}$.  We can compute the quotient $q = q_0 + \cdots + q_{d_n-2}z_n^{d_n-2}$ and remainder $r = r_0 + \cdots + r_{d_n-1}z_n^{d_{n}-1}$ via the linear system generated by $r = c - t_nq$,

\begin{align*}
q_{d_n-2} &= c_{2d_n-2}, \\
q_{d_n-3} &= c_{2d_n-3}  - q_{d_n-2}p_{d_n-1}, \\
q_{d_n-4} &= c_{2d_n-4}  - q_{d_n-3}p_{d_n-1} - q_{d_n-2}p_{d_n-2}, \\
          &\vdots \\
      q_0 &= c_{d_n}  - q_1p_{d_n-1} - \cdots - q_{d_n-2}p_2, \\
r_{d_n-1} &= c_{d_n-1} - q_0p_{d_n-1} - q_1p_{d_n-2} - \cdots - q_{d_n-2}p_1, \\
r_{d_n-2} &= c_{d_n-2} - q_0p_{d_n-2} - q_1p_{d_n-3} - \cdots - q_{d_n-3}p_0, \\
          &\vdots \\
    r_{1} &= c_{1} - q_0p_{1} - q_{1}p_0, \\
      r_0 &= c_0 - q_0p_0.
\end{align*}

We outline a method to solve the equations above. 
The key idea is to compute the entire right-hand-side before reducing by $T_{n-1}$.
This reduces the total number of reductions from quadratic in $d_n$ to linear in $d_n$.
First, set $q_{d_n-2} = c_{2d_n-2}$ and reduce by $T_{n-1}$. 
Then, multiply $q_{d_n-2}p_{d_n-2}$ over $k$ and subtract it from $c_{2d_n-3}$,
and then reduce by $T_{n-1}$ to obtain $q_{d_n-3}$.
It should be clear how to generalize this result and compute all $q_k$. 
Next, to get $r_k$, simply multiply the corresponding $q_ip_j$ over $k$ 
and end by reducing the result of the sum by $T_{n-1}$.
This reveals we only have to do a single reduction per equation 
and each reduction is of a polynomial of degree at most $2(d_j-1)$ in the corresponding variable; 
this takes at most $(2d_n-1)D(n-1)$ field multiplications. 
Multiplying each $q_ip_j$ will take $\delta_{n-1}^2$ field multiplications each. This takes
$$
(1+2+\cdots + (d_n-2))\delta_{n-1}^2 = \binom{d_n-1}{2}\delta_{n-1}^2
$$
field multiplications for computing all $q_i$ in the top $d_n-1$ rows, and
$$
\big((1+2+\cdots + (d_n-1) + (d_n-1)\big)\delta_{n-1}^2 = \left(\binom{d_n}{2}+d_n-1\right)\delta_{n-1}^2
$$
field multiplications for computing all $r_i$ in the bottom $d_n$ rows. Overall,
\begin{align*}
D(n) &= (2d_n-1)D(n-1) + \Big(\frac{(d_n-1)(d_n-2)}{2} + \frac{(d_n(d_n-1)}{2}+d_n-1\Big)\delta_{n-1}^2 \\
&= (2d_n-1)D(n-1) + d_n(d_n-1)\delta_{n-1}^2.
\end{align*}
If there are no extensions, it takes $0$ multiplications to reduce; so we may use $D(0) = 0$ as our initial condition. The solution can be found most easily using Maple's {\tt rsolve} command and some algebraic simplification. The command is

\vspace{5pt}
{\noindent\small\tt > rsolve(\{M(n) = (2*d[n]-1)*M(n-1) +  d[n]*(d[n]-1)*del(n-1)\textasciicircum 2, \\
  \ \ del(n)=del(n-1)*d[n],} {\small\tt M(0)=0, del(1)=d[1]\}, \{M(n), del(n)\});} 
\vspace{5pt}

For the lighter bound, we claim $D(n) \leq 2 \delta_n^2$. To prove this, proceed by induction on $n$. The base case $n=0$ follows from $D(0) = 0 \leq 2 = 2\delta_0$. Next,
\begin{align*}
D(n) &= (2d_n-1)D(n-1) + d_n(d_n-1)\delta_{n-1}^2 \\
     &\leq (2d_n-1)2\delta_{n-1}^2 + d_n(d_n-1)\delta_{n-1}^2 \\
     &= 4d_n\delta_{n-1}^2-2\delta_{n-1}^2 + d_n(d_n-1)\delta_{n-1}^2 \\
     &= 4d_n\delta_{n-1}^2-2\delta_{n-1}^2 + d_n^2\delta_{n-1}^2-d_n\delta_{n-1}^2 \\
     &= 3d_n\delta_{n-1}^2-2\delta_{n-1}^2 + \delta_{n}^2 \\
     &= (3d_n-2)\delta_{n-1}^2 + \delta_{n}^2.
\end{align*}
To finish, note that $3d_n-2 \leq d_n^2$ which follows from $d_n^2-3d_n+2 = (d_n-2)(d_n-1) \geq 0$ for all $d_n\in\Z$. Thus, $D(n) \leq d_n^2\delta_{n-1}^2 + \delta_n^2 = 2\delta_n^2$ and indeed $M(n) \leq \delta_n^2 + D(n) \leq 3\delta_n^2$.
\end{proof}

In \cite{trisetarithmetic}, Li et al prove that multiplication can be done in $O(4^n\delta\log(\delta)\log(\log(\delta)))$ field operations. Their method computes the coefficients by lifting modulo the ideal $\vect{x_n}$ using a Newton-iteration,
and computes the coefficients of $x_1,\dots,x_{n-1}$ recursively. 
It should be noted that the constant in their algorithm is much larger than the one in ours,
so one would expect ours to perform better for smaller degrees.
Comparing these two quantities is not obvious. 
To aid the reader we compare our bound (3) with theirs in the table below.
Because they do not give an explicit constant, we use $3$ since their proofs ensure it is smaller.
The table considers extensions of degree $\delta = d^n$ with $\mdeg(t_i) = d$. 
We give the smallest value of $n$ such that our bound exceeds theirs.

\begin{table}[!htb]
\begin{center}
\begin{tabular}{c|c|c}
\hline
$d$ & $n$ & $\delta = d^n$ \\
\hline
5 & 29 & 186264514923095703125 \\ 
6 & 14 & 78364164096 \\ 
7 & 10 & 282475249 \\ 
8 & 8 & 16777216 \\ 
9 & 6 & 531441 \\ 
10 & 5 & 100000 \\ 
12 & 4 & 20736 \\ 
16 & 3 & 4096 \\ 
28 & 2 & 784 \\ 
115 & 1 & 115 \\ \hline
\end{tabular}
\caption{ \small The first column is the main degree of each $t_i$,
the second is smallest number of extensions where our bound exceeds the bound given in \cite{trisetarithmetic},
the third is the degree of this extension $\delta = d^n$. Values of $d$ that are omitted have the same value of $n$ as the largest shown predecessor. For $d < 5$, our bound is always smaller. For $d \geq 115$, their bound is always smaller.}
\end{center}
\end{table}


Next, we present a field multiplication count for the other arithmetic operations 
we need, namely, division, inversion, and gcd.  
We will not get an exact count as in Proposition \ref{mulcount:mul},
instead focusing on asymptotics. 
We will need these when analyzing the modular gcd algorithm.
We will be using the extended Euclidean algorithm for computing inverses here,
and will only need this result when the field is $\Z_p$.
When using the Euclidean algorithm, we need to assume no zero-divisors are encountered.

\begin{prop}\label{mulcount:remainder}
Let $T\subset k[z_1,\dots,z_n]$ be a triangular set and $R = k[z_1,\dots,z_n]/T$. Let $a,b\in R[x]$ with $\deg(a) \geq \deg(b)$ and $b$ monic. Then the remainder and quotient of $a\div b$ can be computed in $O(\deg(b)(\deg(a)-\deg(b)+1)\delta^2)$ field multiplications.
\end{prop}

\begin{proof}
The standard division works by multiplying the coefficients of $b$ modulo $T$ by an element of $R$ for at most $\deg(a)-\deg(b)+1$ iterations. There are $\deg(b)$ coefficients of $b$ not including the leading coefficient; note that we ignore $\lc(b)$ since we are assuming $b$ is monic. This implies that we need to do $\deg(b)(\deg(a)-\deg(b)+1)$ ring multiplications. We can do ring multiplications in $O(\delta^2)$ field multiplications by Proposition \ref{mulcount:mul}, giving the result.
\end{proof}

\begin{prop} \label{mulcount:inverse}
Let $T\subset k[z_1,\dots,z_n]$ be a triangular set and $R = k[z_1,\dots,z_n]/T$. 
Assume inverses in $k$ can be computed in a $O(1)$ field multiplications. Let $a\in R$. 
Then, assuming no zero-divisors are encountered, $a^{-1}$ can be computed in $O(\delta^2)$ field multiplications.
We use the extended Euclidean algorithm in $a$ and $t_n$ modulo $T_{n-1}$ to compute $a^{-1}$.
\end{prop}

\begin{proof}
Work by induction on $n$. Note that our assumption on $k$ satisfies the base case $n = 0$. Next, let $\delta_{m} = \prod_{i=1}^m \deg(t_i)$ as in Proposition \ref{mulcount:mul}. Let $I(n)$ be the number of field multiplications it takes to compute the inverse of an element with $n$ variables. Then the first step of the Euclidean algorithm is to invert $\lc(a)$. After that, we would have to invert the leading coefficient of the remainder of $t_n \div a$. Since the worst case is the degree of each successive remainder going down by $1$, this will take a total of at most $\deg(a) = \deg(t_n)-1$ recursive inversions. By Proposition \ref{mulcount:remainder}, this will take $O((\deg(t_n)-1)\delta_{n-1}^2)$ field multiplications, the next remainder will take $O((\deg(t_n)-2)\delta_{n-1})$ field multiplications, and so on. In total,
$$
I(n) = \deg(t_n)I(n-1) + \sum_{j=1}^{\deg(t_n)-1}   O(j\delta_{n-1}) = \deg(t_n)I(n-1) + O(\deg(t_n)^2\delta_{n-1}).
$$
Note that we also have to multiply through by the inverse of the leading coefficient at each step. This will take $O(\deg(t_n)^2\delta_{n-1})$ over all steps as well.

Now, the induction hypothesis states $I(n-1) = O(\delta_{n-1}^2)$. So, 
$$
I(n) = \deg(t_n)I(n-1) + O(\delta^2) = \deg(t_n)O(\delta_{n-1}^2) + O(\delta^2) = O(\delta^2),
$$
completing the inductive step. We have not counted the extra multiplications in the extended Euclidean algorithm, but this does not impact the asymptotics; see Theorem 3.11 of \cite{MCA}.
\end{proof}

\begin{prop} \label{mulcount:EA}
Let $T\subset k[z_1,\dots,z_n]$ be a triangular set and $R = k[z_1,\dots,z_n]/T$. Let $a,b\in R[x]$ with $\deg(a) \geq \deg(b)$. Then running the Euclidean algorithm on $a$ and $b$ takes $O(d_ad_b\delta^2)$ field multiplications assuming no zero-divisors are encountered.
\end{prop}

\begin{proof}
Let $d_a = \deg(a)$ and $d_b = \deg(b)$. We will have to perform at most $d_b$ remainders to complete the Euclidean algorithm. This implies we need to invert $d_b$ leading coefficients as well as $\lc(b)$. This accounts for $O(d_b\delta^2)$ field multiplications. Multiplying through by the leading coefficients will cost $O(d_b^2\delta^2) \leq O(d_ad_b\delta^2)$ field multiplications since each remainder has degree $\leq d_b$ and there are $d_b$ of them. Next, computing all but the first remainder cost a total of $O(d_b^2\delta^2)$ field multiplications since each remainder has $\leq d_b$ degree and there are $d_b$ of them in the worst case. Finally, the first remainder costs $O(d_b(d_a-d_b+1)\delta^2)$ field multiplications. Thus, the entire cost is $O(d_ad_b\delta^2 + d_b(d_a-d_b+1)\delta^2) = O(d_ad_b\delta^2)$.
\end{proof}

We will do an asymptotic analysis for the modular c-gcd algorithm that uses Hensel lifting to handle zero-divisors. The running time of the algorithm is dominated by running the Euclidean algorithm modulo multiple primes and the division test. This is verified in the previous section's timing results. Because of this, we will only consider the running time based on these two parts of the algorithm. Also, the expected case is that no zero-divisors are encountered. Further, not encountering a zero-divisor is arguably the worst case scenario. This is because if a zero-divisor is successfully lifted to $\Q$, then the degree of each component will smaller. Therefore, reduction by the triangular set takes less operations. This can also be seen in the timing tests by observing the running time with degrees $[4,4,4]$ and $[4,2,4]$ in Table 5.1 gives a ratio of about $10$.

Now, suppose $M$ primes are needed to successfully compute $g = \gcd(a,b)$. Since there are only finitely many unlucky primes, we assume the algorithm doesn't encounter any of these. This implies we need $M$ runs of the Euclidean algorithm modulo primes. This part takes a total of $O(M\deg(a)\deg(b)\delta^2)$ field multiplications modulo primes by Proposition \ref{mulcount:EA}. We could find a bound on $M$, but we do not think it is worthwhile because our algorithm is output sensitive and any bound will be bad since it has to handle the worst case. Next, the implementation of the algorithm does not perform the division test after each prime. We have coded it so that division is only tested $O(\log(M))$ times. Each division takes $O(\deg(g)\deg(b)\delta^2)$ operations over $\Q$ for a total of $O(\log(M)\deg(g)\deg(b)\delta^2)$ multiplications in $\Q$. Because we're assuming no unlucky primes are encountered, this is an expected case analysis.

We would like to discuss an optimization for the division test. It does not avoid the worst case, but it does improve the expected case. Suppose rational reconstruction successfully outputs a polynomial $h\in\Q[z_1,\dots,z_n]/T~[x]$. Instead of going straight into the division test, we can make use of a check prime. That is, we pick one more prime $p$ where $p$ is not bad or radical. Then, compute $g = \gcd(a,b) \pmod{p}$. If a zero-divisor is encountered in the radical test or in the computation of $g$, we pick a new check prime. Next, we check if $h \equiv g \pmod{p}$. If it is, we proceed to the division test. If it is not we go back to the main loop and pick more primes starting with $p$. More rigorously, we replace lines 33-35 of ModularC-GCD with the following pseudo-code.

\begin{algorithm}
\SetArgSty{text}
Set $h := $ RationalReconstruction($G \pmod M$)\;
\If{$h \ne {\rm FAIL}$}{
    {\bf Check-Prime Loop:} Pick a new prime $p$ that is not bad or radical\;
    \lIf{a zero-divisor is encountered}{pick a new prime, go to Check-Prime Loop}
    Compute $g := \gcd(a,b) {\pmod p}$\;
    \lIf{a zero-divisor is encountered}{pick a new prime, go to Check-Prime Loop}
    \lIf{$g \not\equiv h \pmod{p}$}{pick a new prime, go to Main Loop}
    \lIf{$h \divs a$ and $h \divs b$}{\Return $h$}}
Pick a new prime: Go to Main Loop\;
\end{algorithm}

This optimization only performs the division test once in the expected case. Since there are finitely many unlucky primes by Theorem \ref{thm:unluckyprimes}, the algorithm expects to always pick a lucky prime. Therefore, the only time the division test can be needlessly performed in the expected case, is if not enough primes are picked to exceed the bounded needed by rational reconstruction. The use of a check prime supersedes this since the check prime is expected to be lucky as well. Thus, we have the following theorem.

\begin{theorem}
The ModularC-GCD algorithm performs $O(M\deg(a)\deg(b)\delta^2)$ operations in $\Z_p$. \\
Additionally, it uses $O(\deg(g)\deg(b)\delta^2)$ operations over $\Q$ in the expected case, and \\
$O(\log(M)\deg(g)\deg(b)\delta^2)$ operations in the worst case.
\end{theorem}

\section{Conclusion}

In summary, creating algorithms for computation modulo triangular sets is difficult because of zero-divisors. We have developed the technique of Hensel lifting to resolve this difficulty. We applied this to a modular gcd algorithm that we have shown this gives a practical improvement over the algorithms used in Maple's {\tt RegularChains} package. 
There is room for improvement with our algorithms that should be discussed:
\begin{enumerate}
\item We could avoid the radical prime test as done in the algebraic number field case in \cite{hoeij}.
This would not be a large gain as the radical prime test takes a small fraction of the running time.
\item The division test is a bottleneck of the algorithm and should be the first place to optimize.
We attempted to create a modular division algorithm for this; However, it did not present a gain.
The difficulty is that bounds for the size of the rational coefficients of $\gcd(a,b)$ for in $R[x]$ are too big.
\item Our modular GCD algorithm only works with univariate polynomials over $R$. The obvious way to handle multivariate
polynomials over $R$ would be to use evaluation and interpolation as is done by Brown in \cite{brown} over $\Q$ 
with no extensions. This would require proving results about uniqueness of interpolation over products of fields.
We could also make use of sparse interpolation techniques here, see \cite{hu} and \cite{zippel}.
\end{enumerate}





\end{document}